\documentclass[a4paper,UKenglish,cleveref,autoref,thm-restate]{article}
\usepackage{amsthm}
\usepackage[pdfencoding=auto]{hyperref}
\usepackage{amsmath}
\usepackage[overload]{empheq}
\usepackage{cases}
\usepackage[capitalise,nameinlink,noabbrev]{cleveref}
\usepackage{thm-restate}
\usepackage{amsfonts,amssymb}
\usepackage{enumitem}
\usepackage{graphics}
\usepackage{array}
\usepackage{fullpage}
\usepackage{indentfirst}
\usepackage{xcolor}
\usepackage{tikz}
\usepackage{wrapfig}

\hypersetup{
	colorlinks=true,
	urlcolor=orange,
	linkcolor=blue!50!green,
	citecolor=cyan,
	unicode
}

\newtheorem{theorem}{Theorem}

\newtheorem{lemma}[theorem]{Lemma}
\newtheorem{proposition}[theorem]{Proposition}
\newtheorem{corollary}[theorem]{Corollary}
\newtheorem{claim}[theorem]{Claim}

\theoremstyle{definition}

\newcommand{\mathprob}[1]{\mbox{\textmd{\textsc{#1}}}}
\newcommand{\SAT}{\mathprob{SAT}}

\newcommand{\MCSP}{\mathprob{MCSP}}
\newcommand{\MBPSP}{\mathprob{MBPSP}}

\newcommand{\MFSP}{\mathprob{MFSP}}

\newcommand{\BPIS}{(n\times n)\text{-}\mathprob{BPIS}}

\newcommand{\class}[1]{\mathbf{#1}}

\newcommand{\TIME}{\class{TIME}}
\newcommand{\NP}{\class{NP}}
\newcommand{\coNP}{\class{coNP}}

\newcommand{\CNF}{\class{CNF}}

\renewcommand{\P}{\class{P}}

\newcommand{\XOR}{\mathrm{XOR}}

\newcommand{\BP}{{\rm BP}}
\newcommand{\NBP}{{\rm NBP}}
\newcommand{\onebp}{1\mbox{-}{\rm BP}}
\newcommand{\oneNBP}{1\mbox{-}{\rm NBP}}
\newcommand{\oaBP}{\mathrm{oaBP}}
\newcommand{\DNF}{\mathrm{DNF}}
\newcommand{\OBDD}{{\rm OBDD}}

\newcommand{\defperm}[1]{\kappa_{#1}}

\renewcommand{\phi}{\varphi}

\makeatletter
\let\c@author\relax
\makeatother
\newif\ifdraft
\draftfalse
\ifdraft
\newcommand{\lgl}[1]{\textcolor[rgb]{.00,0.80,0.00}{[LG: #1]}}
\newcommand{\lgs}[1]{\marginpar{\tiny \sf \textcolor[rgb]{.00,0.80,0.00}{[LG: #1]}}}
\newcommand{\ar}[1]{\textcolor[rgb]{.90,0.50,0.50}{[AR: #1]}}
\newcommand{\ars}[1]{\marginpar{\tiny \sf \textcolor[rgb]{.90,0.50,0.50}{[AR: #1]}}}
\else
\newcommand{\lgl}[1]{}
\newcommand{\lgs}[1]{}
\newcommand{\ar}[1]{}
\newcommand{\ars}[1]{}
\fi

\makeatletter
\def\@newl@bel#1#2#3{%
  \@ifundefined{#1@#2}
    {\global\@namedef{#1@#2}{#3}}
    {\gdef\@multiplelabels{}}%
  }
\makeatother

\title{Partial Minimum Branching Program Size Problem \\ is ETH-hard}
\author{Ludmila Glinskih \\ \small{Boston University} \\ \small{\texttt{lglinskih@gmail.com}} \and Artur Riazanov \\ \small{EPFL} \\ \small{\texttt{tunyash@gmail.com}}}
\date{}

\begin{document}

\maketitle

\begin{abstract}
We show that assuming the Exponential Time Hypothesis, the \textsc{Partial Minimum Branching Program Size Problem} ($\MBPSP^{*}$) requires superpolynomial time. This result also applies to the partial minimization problems for many interesting subclasses of branching programs, such as read-$k$ branching programs and $\OBDD$s.

Combining these results with our recent result (Glinskih and Riazanov, LATIN 2022) we obtain an unconditional superpolynomial lower bound on the size of Read-Once Nondeterministic Branching Programs ($\oneNBP$) computing the \emph{total} versions of the minimum $\BP$, read-$k$-$\BP$, and $\OBDD$ size problems.

Additionally we show that it is $\NP$-hard to check whether a given $\BP$ computing a partial Boolean function can be compressed to a $\BP$ of a given size.
\end{abstract}

\section{Introduction}

The {\sc Minimum Circuit Size Problem (MCSP)} has been a central problem of study in computational complexity in recent years. Given the truth-table of a Boolean function, and a size parameter $s$, $\MCSP$ asks one to determine whether the function can be computed by a Boolean circuit of size $s$. As the input length of $\MCSP$ is exponential in the input length of the function defined by the given truth table, it is clear that $\MCSP$ is contained in $\NP$. Although it is not known whether $\MCSP$ is $\NP$-complete, it is widely believed to be hard to compute. Starting from the work of Kabanets and Cai \cite{KabanetsC00} there has been mounting evidence that $\MCSP \not\in \P$ \cite{RazborovR97}. For example, $\MCSP \in \P$ implies that widely believed cryptographic assumptions do not hold. On the other hand, there are complexity-theoretic barriers to proving $\NP$-hardness of $\MCSP$ \cite{MW17, SaksS20}.

Even before the work of Kabanets and Cai \cite{KabanetsC00}, variations of $\MCSP$ were studied in the context where the function is given not as a truth table, but represented in some computational model. One family of models that was extensively studied in this context is \emph{branching programs}. The line of work \cite{BolligW96, OliveiraCVS98, TakenagaY00, Sieling02, Bollig16} showed that branching program minimization is $\NP$-hard for the natural cases of succinct representation of a function: either the function itself is represented as a branching program as in \cite{BolligW96, Sieling02, Bollig16}, or an input is a \emph{partial function} with few fixed values as in \cite{TakenagaY00}. The results, however, are limited to weak models of branching programs such as $\OBDD$s or read-once branching programs. To our knowledge, no hardness results are known even for minimizing read-twice branching programs.

The potential cross-pollination between these two research areas motivates us to study {\sc Minimum Branching Program Size Problem} ($\MBPSP$): given a function represented as a truth table $\{0,1\}^{2^n}$ and a size parameter $s \in [2^n]$ determine whether exists a \emph{branching program} of size at most $s$ computing the given function.

On the meta-complexity side, one of the approaches to prove $\NP$-hardness of $\MCSP$ is to show lower bounds for $\mathcal{C}$-$\MCSP$ problems, in which we ask whether an input Boolean function can be computed by a circuit from the class $\mathcal{C}$, where $\mathcal{C}$ is a weaker class than general Boolean circuits. There are already results showing $\NP$-hardness of $\DNF$-$\MCSP$ \cite{Masek79}, $(\DNF \circ \XOR)$-$\MCSP$ \cite{HiraharaOS18}, and the hardness of DeMorgan Formulas-MCSP ($\MFSP$) assuming the Exponential Time Hypothesis (ETH) \cite{Ilango21}. All of these reductions share a common structure.
\begin{itemize}[noitemsep]
    \item
First they show hardness of the Partial $\mathcal{C}$-$\MCSP$ ($\mathcal{C}$-$\MCSP^{*}$) problem in which the input is the truth table of a partial function.
    \item Then they show a reduction from the $\mathcal{C}$-$\MCSP^{*}$ problem to the $\mathcal{C}$-$\MCSP$ problem.
\end{itemize}
Hardness results for Partial MCSP ($\MCSP^{*}$) are already known. In \cite{Ilango20} Ilango showed that $\MCSP^{*}$ requires superpolynomial time to compute\footnote{The same hardness result holds for Partial MFSP ($\MFSP^{*}$).}, assuming the Exponential Time Hypothesis. Recently, Hirahara \cite{Hirahara22} showed the $\NP$-hardness of $\MCSP^{*}$ under randomized reductions. So the only missing component for showing hardness of $\MCSP$ is to show an efficient reduction from $\MCSP^{*}$ to $\MCSP$.

The study of hardness of $\MBPSP$ under the ETH continues the line of work above. Expressive power of branching programs is situated between that of Boolean formulas and Boolean circuits \cite{Pratt75, Giel2000}. On the other hand our understanding of branching programs seems to be slightly better than our understanding of circuits. One evidence for this is that the strongest lower bound on the $\BP$ size for an explicit function is $\Tilde{\Omega}(n^2)$ \cite{Nec66} which is much stronger than $\Omega(n)$ lower bound for Boolean circuits \cite{FindGHK16}. Showing the ETH-hardness of $\MBPSP$ would improve the recent ETH-hardness result for $\MFSP$ \cite{Ilango21} which is the frontier hardness result for total circuit minimization problems and move us one step closer to the ETH-hardness of $\MCSP$ itself. Our main result is completing the first step of the common recipe for establishing the hardness of $\mathcal{C}$-$\MCSP$ by showing hardness of $\BP\text{-}\MCSP^*$ (that we call $\MBPSP^{*}$ for short).

\begin{theorem}\label{theorem:main-theorem}
Every deterministic algorithm computing $\MBPSP^{*}$ on inputs of length $N$ requires time $N^{\Omega(\log\log(N))}$ assuming the Exponential Time Hypothesis.
\end{theorem}

 Note that the ETH-hardness of $\MCSP^*$ or $\MFSP^*$ does not imply the ETH-hardness of $\MBPSP^*$. Although on surface it seems that $\MBPSP$ ($\MBPSP^*$) is just a special case of $\MCSP$ ($\MCSP^*$), the hardness of the latter does not necessarily imply the hardness of the former. For example, in some special settings of parameters, an explicit construction of a many-one reduction between $\MCSP$ and $\MBPSP$ yields breakthrough circuit or $\BP$ lower bounds.\footnote{Suppose that $f$ is such a reduction from $\MCSP_{s=3n}$ (the language of the truth-tables computable by circuits with at most $3n$ gates) to $\MBPSP_{s=n^{2.1}}$. Then taking any function requiring circuits of size larger than $3n$ \cite{FindGHK16, LY22} and applying $f$ to it yields $n^{2.1}$ lower bound for general branching programs superior to the current state of the art. In a similar way, we get superlinear circuit lower bounds from any reduction from $\MBPSP_{s=n^2}$ to $\MCSP_{s=n^{1+\varepsilon}}$.} On the other hand, ruling out more natural parameter configurations does not seem as trivial: indeed, our main result shows that the identity function is a reduction on a certain subset of the instances.

\medskip

On the branching program minimization side, the consequences are twofold. First, our result implies that all partial branching program minimization problems require time $N^{\Omega(\log \log N)}$ under ETH for any representation of a function that is at least as efficient as truth table.\footnote{We note that results in this area often establish $\NP$-hardness and even hardness of size approximation for the case of concise representation, while our result only establishes ETH-hardness of the exact version of the problem.} Second, we show that our result implies $\NP$-hardness of minimization of a branching program computing a partial function that itself is represented as a branching program (see \cref{sec:bp-min} for details).

In particular, our result implies that the problem of finding the smallest $\OBDD$ (a subclass of branching programs) that computes the given partial function is ETH-hard. But surprisingly, there is a polynomial-time algorithm that finds the smallest $\OBDD$ computing a \emph{total} function \cite{FS90}. This means that assuming ETH holds, the potential reduction from $\MBPSP^*$ to $\MBPSP$ must fail on $\OBDD$s.

\subsection{Techniques}

To show \cref{theorem:main-theorem} we adapt the proof of a similar result of Ilango \cite{Ilango20} for the hardness of $\MCSP^*$ and $\MFSP^*$ problems.

\begin{theorem}\label{theorem:Ilango}\cite{Ilango20}
Every deterministic algorithm computing $\MCSP^{*}$ on inputs of length $N$ requires time $N^{\Omega(\log\log(N))}$ assuming the Exponential Time Hypothesis holds. The same holds for $\MFSP^*$.
\end{theorem}

Ilango's proof of \cref{theorem:Ilango} is a reduction from the $(n \times n)$-Bipartite Permutation Independent Set Problem ($\BPIS$). For now we can think of it as an independent set problem on some special set of graphs. This is an $\NP$-hard graph problem introduced in \cite{LMS18}, where the authors show that any deterministic Turing machine requires $2^{\Omega(n \log n)}$-time\footnote{The input to this problem has $\Theta(n^4 \log n)$ bit size.} to solve it if the Exponential Time Hypothesis holds.
 For an input graph $G$ with vertices labeled by elements of $[n] \times [n]$ Ilango constructs a truth-table of a $3n$-variable partial Boolean function $\gamma_G$, such that
 \begin{enumerate}[label=($\star$)]
     \item $G$ is a {\sc Yes}-instance of $\BPIS$ iff $\gamma_G$ can be computed by a DeMorgan formula with at most $3n-1$ gates. \label{item:last-step}
\end{enumerate}

The key point here is that if we constrain a circuit or a formula to have at most $3n-1$ (inner) gates, its structure becomes simple enough for a very precise analysis. \cite{Ilango20} characterizes the structure of such formulas that compute $\gamma_G$ which would yield the correspondence between such formulas with independent sets in $G$.

Although branching programs are situated between formulas and circuits in computational expressivity \cite{Giel2000}, \cref{theorem:Ilango} does not apply to $\MBPSP^*$, since with a sharp constraint on the size, branching programs may become stronger than circuits. We observe in \cref{prop:oabp-vs-circuits} that this is indeed the case, showing that branching programs are stronger than circuits if we constrain the size parameter of both models to ``minimal viable'' size: $m-1$ inner nodes for a DeMorgan circuit and $m$ non-sink nodes for a $\BP$ for the case of a Boolean function on $m$ input bits.

Nevertheless, we show that the reduction to $\MCSP^{*}$ from \cref{theorem:Ilango} does yield ETH-hardness for $\MBPSP^*$. Indeed, in our proof of \cref{theorem:main-theorem} we only change the part \ref{item:last-step}, keeping the construction of $\gamma_G$. The key part of the proof is to understand what structure a branching program has if it has at most $3n$ (non-sink) nodes. As $\gamma_G$ is sensitive in every input variable, we need to understand a structure of a branching program that has unique variable labels. We call such programs  \textit{once-appearance branching programs} ($\oaBP$). We show that $\gamma_G$ can be computed by a once-appearance branching program iff $G$ is a {\sc Yes}-instance of $\BPIS$.

\subsection{Corollaries}
\subsubsection{Weaker Classes}
$\BP$ is not the only model that collapses to $\oaBP$ when constrained to have at most $n$ non-sink nodes. The same also holds for $\OBDD$, and $k$-$\BP$ for every positive integer $k$. From this we get a family of hardness of minimization results for restricted classes of branching programs.
\begin{corollary}
Every deterministic algorithm computing $\mathcal{C}$-$\MCSP^*$ for $\mathcal{C} \in \{\OBDD, \BP\} \cup \{k\text{-}\BP \mid k \ge 1\}$\footnote{Or any other class of branching programs that collapses to $\oaBP$ on the $n$-node programs.} on inputs of length $N$ requires time $N^{\Omega(\log\log(N))}$ assuming the Exponential Time Hypothesis.
\end{corollary}

\subsubsection{Unconditional Results for a Weaker Model}
In our recent work \cite{GlinskihR22} we employed the framework of Ilango from \cite{Ilango20} to show an unconditional lower bound on the size of the read-once nondeterministic branching program ($\oneNBP$) computing the total version of $\MCSP$. The proof consists of three components.
\begin{enumerate}[noitemsep]
    \item Show that $\BPIS$ unconditionally cannot be computed by a $\oneNBP$ of smaller than $\Omega(n!)$ size.
    \item Show that Ilango's reduction from $\BPIS$ to $\MCSP^*$ can be efficiently encoded using $\BP$. This shows that if $\MCSP^{*}$ on inputs of length $N$ could be computed by a $\oneNBP$ of size $N^{o(\log\log(N))}$, then $\BPIS$ could be computed by a $\oneNBP$ of size $o(n!)$.
    \item Show that $\MCSP$ and $\MCSP^{*}$ have the same $\oneNBP$ complexity up to a linear multiplicative factor. This implies that $\oneNBP$ complexity of $\MCSP$ is at least $N^{\Omega(\log\log(N))}$.
\end{enumerate}
Since the proof of \cref{theorem:main-theorem} follows similar steps as the proof of the \cref{theorem:Ilango}, and analogs of the third property from \cite{GlinskihR22} also hold for $\BP$, $k$-$\BP$, and $\OBDD$, we get similar unconditional $\oneNBP$ lower bounds for total minimization problems in all these models.
\begin{corollary}
    Every $\oneNBP$ computing $\mathcal{C}$-$\MCSP$ for $\mathcal{C} \in \{\OBDD, \BP\} \cup \{k\text{-}\BP \mid k \ge 1\}$ on inputs of length $N$ requires size at least $N^{\Omega(\log\log(N))}$.
\end{corollary}

\subsubsection{\texorpdfstring{$\NP$}{NP}-hardness of Compressing BPs}
\label{sec:bp-min}
Additionally, we show that given a graph $G$ we can construct a branching program $B$ of polynomial size that exactly computes the partial function  $\gamma_G$. Namely it outputs $\{0,1\}$ values on inputs on which $\gamma_G$ is defined and ${*}$ on all other inputs. Then $B$ can be compressed to an $\oaBP$ which agrees with $\gamma_G$ on its domain iff $G$ is a {\sc Yes}-instance of $\BPIS$. Combining this result with a folklore result of $\coNP$-hardness of compressing a branching program, we get the following result on the hardness of compressing branching programs.

\begin{corollary}
The language of pairs $(B, s)$ where $B$ is a branching program\footnote{The same is true for all sub-classes of branching programs that simulate $4$-$\BP$.} over the alphabet $\{0,1,*\}$ and $s$ is in integer, such that an extension of the partial function described by $B$ can be computed by a $\BP$ of size at most $s$, is $\NP$ and $\coNP$-hard.
\end{corollary}

\subsection{Organization}
The structure of the paper is as follows:
\begin{itemize}[noitemsep]
    \item In \cref{section:prelims} we give definitions of the branching programs and various versions of $\MCSP$ we are working with.
    \item In \cref{section:main} we prove our main result: the hardness of $\MBPSP^*$ under ETH (\cref{theorem:main-theorem}).
    \item In \cref{section:corollaries} we sketch the proofs of several corollaries of \cref{theorem:main-theorem}.
    \item  Finally, in \cref{section:future-work} we explain the future directions of this work and possible barriers we see for strengthening our results.
\end{itemize}

\section{Preliminaries}\label{section:prelims}
Throughout this work we denote by $[n]$ the set of $n$ elements $\{1,\dots,n\}$.
Let $a$ and $b$ be partial assignments. If the supports of $a$ and $b$ do not intersect, denote by $a \cup b$ the partial assignment that coincides with $a$ on the support of $a$ and with $b$ on the support of $b$. For a graph $G$ we denote its set of vertices by $V(G)$ and its set of edges by $E(G)$. Throughout the paper a \emph{circuit} is a DeMorgan circuit ($\land$, $\lor$ and $\lnot$ gates) with gates of arity 2. A \emph{formula} is a tree-like DeMorgan circuit. A \emph{read-once formula} is a formula such that each input bit feeds into exactly one gate. E.g. $x \lor (\lnot y \land z)$ is a read-once formula, and $(x \land y) \lor (\lnot x \land z)$ is not. Since every DeMorgan formulas can be transformed in such a way that all $\lnot$-gates are adjacent to the inputs, we identify DeMorgan formulas with trees with \emph{literals} in the leaves and $\land$/$\lnot$ in the inner nodes.

A \emph{Branching Program} ($\BP$) is a form of representation of functions. A $n$-ary function\footnote{We will instantiate this definition for $D = \{0,1\}$ and $D = \{0,1,*\}$} $f\colon D^n \to \{0,1\}$, where $|D|=d$, is represented by a directed acyclic graph with exactly one source and two sinks. Each non-sink node is labeled with a variable; every internal node has exactly $d$ outgoing edges: labeled with elements of $D$. One of the sinks is labeled with $1$ and the other is labeled with $0$. We say that a node \emph{queries} $x$ if it is labeled with a variable $x$. The \emph{size} of a branching program is a number of nodes in it.

The value of the function for given values of the variables is evaluated as follows: we start in the source and then follow the edge labeled with the value of the variable in the source, we repeat this process until the current node is a sink. The label of this sink is the value of the function. We refer to a path that terminates in a $\sigma$-labeled sink as \emph{$\sigma$-path}.

A \emph{Non-deterministic Branching Program} ($\NBP$) is a Branching Program over the alphabet $\{0,1\}$ with additional non-labeled \emph{guessing} nodes with two non-labeled outgoing edges. The value computed by a Non-deterministic Branching Program is $1$ if there exists at least one $1$-path that agrees with the given assignment.

A Branching Program (deterministic or not) is called \emph{read-once} (denoted as $\onebp$ in the deterministic case and $\oneNBP$ in the non-deterministic case) if any path from the source to a sink contains each variable label at most once.

For a partial assignment $\alpha\colon [n] \to \{0,1,*\}$ and a branching program $B$ with $n$ input bits over the alphabet $\{0,1\}$ we write $B|_{\alpha}$ for a branching program obtained by removing every edge of form $x_i \neq \alpha(i)$ and contracting every edge of form $x_i = \alpha(i)$ if $\alpha(i) \in \{0,1\}$. It is easy to see that if $B$ computes the function $f$, then $B|_{\alpha}$ computes the function $f|_{\alpha}$.

The \textsc{Minimum Branching Program Size Problem} ($\MBPSP(f,s)$) gets as input the truth-table of a Boolean function $f\colon \{0,1\}^n \to \{0,1\}$ of length $N=2^n$ and a size parameter $s$, and outputs $1$ if there exists a branching program of size at most $s$ that computes $f$. We use notation $\MBPSP_{s'(n)}(f) = \MBPSP(f, s'(n))$.

The \textsc{Partial Minimum Branching Program Size Problem} ($\MBPSP^{*}(f^{*},s)$) gets as input the partial truth-table of a function $f^{*}\colon \{0,1\}^n \to \{0,1, *\}$ of length $N=2^n$ and a size parameter $s$, and outputs $1$ if there exists a substitution of each $*$ in the truth-table to a $0/1$, transforming $f^{*}$ to a Boolean function $f$, such that $\MBPSP(f,s)=1$.
We use notation $\MBPSP^*_{s'(n)}(f) = \MBPSP^*(f, s'(n))$.

\subsection{Once-Appearance Branching Programs}
A \emph{once-appearance branching program} ($\oaBP$) is a branching program with distinct node labels. It is easy to see that for a function depending on $n$ variables such $\oaBP$ representation contains at most $n$ non-sink nodes.
For $\oaBP$ we identify the non-sink nodes with their labels, and identify the edges with equalities of form $x = \alpha$ for a variable $x$ and $\alpha$ from the alphabet.

It is easy to see, that $\oaBP$s are strictly stronger than read-once DeMorgan formulas.
\begin{proposition}
\label{prop:de-morgan-to-oabp}
    If $f$ is computable by a read-once DeMorgan formula, then it is computable by an $\oaBP$.
\end{proposition}
\begin{proof}
    We first observe that if $g$ and $h$ are computable by $\oaBP$s, then so are $g\land h$ and $g \lor h$ whenever $g$ and $h$ have disjoint set of variables. Indeed $\oaBP$ for $g\land h$ is constructed by redirecting all edges going to the $1$-in the $\oaBP$ computing $g$ to the source of the $\oaBP$ computing $h$, and by redirecting all edges to the $0$-sink in $g$ to the $0$-sink of $h$. For $g \lor h$ we redirect all edges to the $0$-sink of to the source of the $\oaBP$ for $h$, and all edges to the $1$-sink to the $1$-sink of $h$. Now to finish the proof observe that all literals $\{x_1, \lnot x_1, \dots, x_n, \lnot x_n\}$ are trivially computable by an $\oaBP$.
\end{proof}
\begin{proposition}
\label{prop:oabp-vs-circuits}
    Let $g(x,y,z) \coloneqq (y \land x) \lor (z \land \lnot x)$.
    Then $g$ is computable by an $\oaBP$ but is is not computable by a read-once DeMorgan formula.
\end{proposition}
\begin{proof}
    The $\oaBP$ for $g$ is a decision tree with the root $x$, the edge $x=0$ going to $z$, the edge $x=1$ going to $y$, the edge $y=\alpha$, $z=\alpha$ going to the $\alpha$-sink for $\alpha \in \{0,1\}$.

    Consider a read-once formula $G$ computing $g$. Then $G|_{y=0, z=1}$ computes a function $\lnot x$, so $G$ must have a leaf  $\lnot x$. But $G|_{y=1,z=0}$ computes a function $x$, so $G$ must have a leaf $x$. That contradicts the read-once property of $G$.
\end{proof}
We note that this example can be extended to a function that is sensitive in an arbitrary number of variables, say for $g(x,y,z) \lor t_1 \dots \lor t_n$.

\subsection{Bipartite Permutation Independent Set Problem}
In \cite{Ilango20} Ilango showed that Partial $\MCSP$ is ETH-hard. The main idea of the proof is to show a reduction from the $(n \times n)$-Bipartite Permutation Independent Set problem that was previously shown to be ETH-hard in \cite{LMS18}. In this section and later we assume that $n$ is even.

In {\sc $(n \times n)$-Bipartite Permutation Independent Set problem} ($\BPIS$) we are given the adjacency matrix of an undirected graph $G$ over the vertex set $[n] \times [n]$ where every edge is between the sets of vertices $J_1=[n/2]^2$ and $J_2=\{n/2+1, \dots, n\}^2$. A graph $G$ is a {\sc Yes}-instance iff it contains an independent set $S \subseteq J_1 \cup J_2$ of size $n$ such that the coordinates of vertices in $S$ define a permutation of $[n]$, i.e. $\forall i \in [n] \: \exists j,k \in [n]\colon v=(i,j), w=(k,i), v, w \in S$.

Note that a permutation of $[n]$ corresponding to a {\sc Yes}-instance of $\BPIS$ can be viewed as two permutations on disjoint $n/2$-element subsets. One permutation is defined by $n/2$ vertices from $J_1$ (corresponding to a permutation of elements $[n/2]$), and another by $n/2$ vertices from $J_2$ (corresponding to a permutation of elements $\{n/2+1,\dots,n\}$). We use each of these interpretations interchangeably.

\begin{theorem}[\cite{LMS18}]
\label{theorem:eth-hardness-of-bpis}
{\sc $(n \times n)$-BPIS} cannot be solved in deterministic time $2^{o(n\log{n})}$ unless ETH fails.
\end{theorem}

Ilango in \cite{Ilango20} gave a $2^{O(n)}$-time reduction from {\sc $(n \times n)$-BPIS} to $\MCSP$*, hence, showing that $\MCSP$* cannot be solved in deterministic time $2^{o(n\log{n})}$ unless ETH fails.

\section{Proof of the Main Theorem}\label{section:main}
In this section we show that the reduction Ilango constructed from ETH-hard problem $\BPIS$ to a formula minimization problem holds for the $\BP$ minimization problem as well. More precisely we show:
\begin{theorem}
\label{theorem:reduction-from-bpis}
Suppose that $\MBPSP^{*}\colon\{0,1, *\}^{2^n} \times [2^n] \to \{0,1\}$ is in $\TIME[N^{o(\log \log (N)}]$ where $N=2^n$, then the Exponential Time Hypothesis is false.
\end{theorem}

Let us recall the reduction $\mathcal{R}\colon \{0,1\}^{\binom{n \times n}{2}} \to \{0,1,*\}^{2^{3n}} \times [2^{3n}]$ from \cite{Ilango20}\footnote{We ignore all the edges except for ones between the sets $[n/2] \times [n/2]$ and $\{n/2+1,\dots n\} \times \{n/2+1,\dots,n\}$, the second argument is represented in binary form.}.
Let $G = ([n] \times [n], E)$ be an instance of $\BPIS$. Then the reduction outputs the pair $(t,3n-1)$ where $t$ is the truth table of a partial function $\gamma_{G}\colon \{0,1\}^{n} \times \{0,1\}^{n} \times \{0,1\}^{n} \to \{0,1,*\}$ defined as

\newsavebox{\gammadef}
\sbox{\gammadef}{%
\parbox{\textwidth}{%
\begin{empheq}[left = {\gamma_{G}(\vec{x},\vec{y},\vec{z}) = \empheqlbrace}]{align}
       &\textstyle\bigvee_{i \in [n]} (y_i \land z_i), &&\quad \text{if } \vec{x}=0^{n} \label{item:x-is-0} \\
       & \textstyle\bigvee_{i \in [n]} z_i, &&\quad \text{if } \vec{x}=1^{n} \label{item:x-is-1} \\
       & \textstyle\bigvee_{i \in [n]} (x_i \lor y_i), &&\quad \text{if } \vec{z}=1^{n} \label{item:z-is-1} \\
       & 0, &&\quad \text{if } \vec{z}=0^{n} \label{item:z-is-0}\\
       & \textstyle\bigvee_{i=1}^{n/2} x_i &&\quad \text{if } \vec{z}=1^{n/2}0^{n/2} \text{ and } y=0^{n} \label{item:y-0-z-1-0}\\
       & \textstyle\bigvee_{i=n/2+1}^{n} x_i &&\quad \text{if } \vec{z}=0^{n/2}1^{n/2} \text{ and } y=0^{n} \label{item:y-0-z-0-1} \\
       & 1 &&\quad \text{if } \exists ((j,k), (n + j',n + k')) \in E\colon (\vec{x},\vec{y},\vec{z}) = (\overline{e_k e_{k'}}, 0^{n}, e_j e_{j'}) \label{item:edges} \\
       & * &&\quad \text{otherwise}
    \end{empheq}}}
\noindent\usebox{\gammadef}
where $\vec{x}=(x_1,\dots,x_{n})$, $\vec{y}=(y_1,\dots,y_{n})$, $\vec{z}=(z_1,\dots,z_{n})$, and $e_i \in \{0,1\}^{n/2}$ is the vector with 1 in the $i$th entry and zeroes in all the others, and $\overline{e_i}$ is such that $e_i + \overline{e_i} = 1^{n/2}$. In the rest of the proof, we use $e_i$ as a vector with the $i$th coordinate equal to $1$ and all the rest equal to zero with the appropriate dimension.

The intuition behind this reduction is that for every read-once formula, computing this function has to group triples of bits $x_i, y_j,$ and $z_k$ such that $j=k$, and $i \to j$ encodes that an element $i$ maps to an element $j$ in the permutation. The cases \eqref{item:x-is-0}-\eqref{item:z-is-0} ensure that the read-once formula encodes some permutation on $[n]$. The cases \eqref{item:y-0-z-1-0} and \eqref{item:y-0-z-0-1} ensure that the permutation maps the first $n/2$ elements into the first $n/2$ elements and the last $n/2$ elements into the last $n/2$ elements. And, finally, the case \eqref{item:edges} enforces that for every edge in the original graph, the permutation defined by the read-once formula contains at most one endpoint of this edge. Note, that every endpoint or vertex chosen in the independent set encodes a mapping of one element of $[n]$ to one element of $[n]$ in the permutation.

As was shown in \cite{Ilango20}, this reduction maps {\sc Yes}-instances of BPIS to {\sc Yes}-instances of $\MCSP^{*}_{s=m-1}$ where $m=3n$ is the number of input bits of $\mathcal{R}(G)$:
\begin{theorem}[\cite{Ilango20}]
\label{theorem:reduction-correctness-circuits}
For the reduction $\mathcal{R}$ as above  $G \in \BPIS$ iff $\mathcal{R}(G) \in \MCSP_{s=m-1}^{*}$.
\end{theorem}

To prove \cref{theorem:reduction-correctness-circuits}, Ilango shows that $\gamma_G$ extends to a total function computable by a \textit{read-once monotone formula} iff $G$ is a {\sc Yes}-instance of $\BPIS$. In our work we show that a similar statement holds for \textit{once-appearance branching programs} ($\oaBP$).

\begin{theorem}
\label{theorem:reduction-correctness-oaBP}
For the reduction $\mathcal{R}$ as above $G \in \BPIS$ iff $\mathcal{R}(G)$ can be computed by $\oaBP$.
\end{theorem}
The plan of the proof is the following:
\begin{itemize}[noitemsep]
    \item In \cref{sec:gamma-prime-structure} we show that the topological order of an $\oaBP$ computing $\gamma_G$ defines a unique permutation $\pi\colon [n] \to [n]$.
    \item In \cref{sec:gamma-G-structure} we show that for $\pi$ as above $\{(i, \pi(i)) \mid i \in [n]\}$ is an independent set in $G$, thus showing the ``only if'' direction of the equivalence.
    \item Since by \cref{prop:de-morgan-to-oabp} that $\oaBP$ losslessly simulates read-once formulas, we conclude that \cref{theorem:reduction-correctness-circuits} implies the ``if'' direction of the equivalence.
\end{itemize}

For the first step of our plan it is convenient to simplify the definition of $\gamma_G$.
Let $\gamma'_n\colon \{0,1\}^{3n} \to \{0,1,*\}$ be equal to $\gamma_{G}$ whenever one of \eqref{item:x-is-0}-\eqref{item:z-is-0} is satisfied, and $*$ otherwise. Here we may choose an arbitrary $G$ as in this domain $\gamma_G$ does not depend on it.
\begin{multline*}
    \gamma'_n(\vec{x}, \vec{y}, \vec{z}) =
    \begin{cases}
        \bigvee_{i \in [n]} (y_i \land z_i), & \text{if } \vec{x}=0^{n} \\
        \bigvee_{i \in [n]} z_i, & \text{if } \vec{x}=1^{n} \\
        \bigvee_{i \in [n]} (x_i \lor y_i), & \text{if } \vec{z}=1^{n} \\
        0, & \text{if } \vec{z}=0^{n} \\
        * & \text{otherwise}
    \end{cases}
\end{multline*}
We usually shorthand $\gamma'_n$ as $\gamma'$.

\subsection{Structure of \texorpdfstring{$\oaBP$}{oaBP} Computing \texorpdfstring{$\gamma'$}{gamma prime}}
\label{sec:gamma-prime-structure}

In this section we identify the structure of an $\oaBP$ computing $\gamma'$. Let us start by constructing an example of such an $\oaBP$. As we mentioned before, whenever a function is computable by a read-once formula, it is computable by $\oaBP$. Ilango \cite{Ilango20} shows that all read-once formulas computing $\gamma'$ have form
\begin{equation}
\label{eq:read-once-formula}
    \bigvee_{i\in [n]} x_{\pi(i)} \lor (y_i \land z_i)
\end{equation} for a permutation $\pi \in S_n$. One way to translate this into an $\oaBP$ is to chain together $n$ three-node $\oaBP$s each of them computing one disjunct of \eqref{eq:read-once-formula}. A topological order of such a $\oaBP$ is a sequence of $n$ triplets of $\{x_{\pi(i)}, y_i, z_i\}$ with variables in triplets permuted in some order. We claim that topological order of \emph{every} $\oaBP$ computing $\gamma'$ has this form\footnote{Notice that we do not show that the $\oaBP$ in question has the exact form as the one constructed according to \eqref{eq:read-once-formula}.}.
\begin{lemma}
\label{lem:order}
Suppose $B$ is an $\oaBP$ that computes $\gamma'$. Then let $\ell_1, \ell_2, \dots, \ell_{3n}$ be labels of vertices of $B$ in a topological order. Then for every $i \in \{0,..,n-1\}$ we have $\{\ell_{3i + 1}, \ell_{3i + 2}, \ell_{3i + 3}\} = \{x_a,y_b, z_b\}$ for some $a,b \in [n]$.
\end{lemma}

We then say that $B$ \emph{defines} a permutation $\defperm{B}$ mapping $b \mapsto a$ iff $\{x_a, y_b, z_b\} = \{\ell_{3i + 1}, \ell_{3i + 2}, \ell_{3i + 3}\}$ for some $i \in \{0, \dots, n-1\}$. We refer to the triplets of consecutive nodes of form $\{x_a, y_b, z_b\}$ as \emph{$xyz$-triplets} or \emph{$xyz$-blocks}.

The proof of \cref{lem:order} is surprisingly delicate. The starting point is the characterization of the structure of $\oaBP$s computing the subfunctions of $\gamma'$ corresponding to the conditions \eqref{item:x-is-0}-\eqref{item:z-is-0}. The main tool in the proofs of these characterizations is the reduction to query complexity, e.g. $\gamma'|_{\vec{x}=0^n}$ has maximum possible query complexity $2n$, which enforces a very tight constraint on the structure of the $\oaBP$ computing it. A potential easy proof of \cref{lem:order} could have followed a similar logic as the proof above if $\gamma'$ had maximal query complexity $3n$. This, however, is not the case, query complexity\footnote{Query $x_1, x_2, z_1, z_2$, if $x_1 \neq x_2$ or $z_1 \neq z_2$, then we may return any value, as none of the conditions \eqref{item:x-is-0}-\eqref{item:z-is-0} hold. Otherwise, for every possible combination of values $x_1, z_1$ we have only two possible conditions among \eqref{item:x-is-0}-\eqref{item:z-is-0} that might hold. If $x_1 = z_1 = \alpha$ we return $\alpha$ right away, if $x_1 = 0$ and $z_1=1$, it suffices to query $\vec{y}$ and $\vec{z}$ to get the answer, if $x_1 = 1$ and $z_1 = 0$, it suffices to query $\vec{x}$ and $\vec{z}$ to get the answer.} of $\gamma'$ is at most $2n+2$, so such a simple argument is unfortunately ruled out.

\subsubsection{Structure of \texorpdfstring{$\gamma'$}{gamma prime} Conditioned on \texorpdfstring{\eqref{item:x-is-0}-\eqref{item:z-is-0}}{conditions}}
We start by characterizing the structure of $\oaBP$ that computes $\gamma'|_{\vec{x}=0^n}$.
\begin{lemma}
\label{lem:or-of-ands-structure}
The structure of an $\oaBP$ computing the function $\bigvee_{i \in [n]} (y_i \land z_i)$ is the following: all nodes except for the $1$-sink form a single $2n$-edge path  $(a_1, b_1, a_2, b_2, \dots, a_n, b_n, 0\text{-sink})$ where $\{a_i, b_i\} = \{y_{\pi(i)}, z_{\pi(i)}\}$ for some permutation $\pi \in S_n$. For each $i \in [n]$, the edge $a_i = 1$ goes to $b_i$, the edge $b_i = 0$ goes to $a_{i+1}$, the edge $a_i = 0$ goes to $a_{i+1}$ (for $a_n$ it goes to the $0$-sink), and the edge $b_i = 1$ goes to the $1$-sink.
\end{lemma}
\begin{proof}
The query complexity of the function $\bigvee_{i \in [n]} (y_i \land z_i)$ is $2n$, hence there should exist at least one path of length $2n$ in this branching program. Next, consider the first two variables queried in this path. Assume they do not correspond to $y$ and $z$ with the same label. Without loss of generality assume the first node is labeled by $y_1$ and the second node is labeled by $z_2$ (observe that the cases of the first two labels being $y_1$ and $y_2$ or $z_1$ and $z_2$ are identical, since the variables in the pairs are symmetric). Now, both edges outgoing from $y_1$ should be parts of the paths that lead to $z_2$. As for both values of $y_1=0$ and $y_1=1$ exists a $1$-assignment that is zero in all variables other than $z_2 = y_2 = 1$. Hence, if after querying $y_1$ we do not query $z_2$, we get a BP that incorrectly computes the function. As we consider the first two nodes of the longest path, we get that both edges from $y_1$ should directly go to the node labeled by $z_2$. But it turns the $\BP$ to be insensitive to a change in the value of $y_1$ that contradicts the sensitivity of $\bigvee_{i \in [n]} (y_i \land z_i)$ function. Hence the edge $y_1=1$ must lead to the node labeled with $z_1$. Then it remains to observe that after applying the substitutions $y_1 = 0$ or $y_1 = 1, z_1 = 0$ the query complexity of the function decreases by exactly $2$. Note that paths corresponding to these two substitutions must lead to the third node in the branching program. This node then must compute $\bigvee_{i \in [n] \setminus \{1\}} (y_i \land z_i)$, so the rest of the $\oaBP$ has the desired structure by the simple induction on $n$.
\end{proof}
Next, we summarize the structures for the conditions \eqref{item:x-is-1}-\eqref{item:z-is-0}.
\begin{lemma}
\label{lem:structures}
The following structural properties hold.
\begin{itemize}[noitemsep]
    \item[\eqref{item:x-is-1}] An $\oaBP$ computing the function $\gamma'_n|_{\vec{x}=1^n} (\vec{y}, \vec{z}) = \bigvee_{i \in [n]} z_i$ has the following property. Every node $v$ of this $\oaBP$ reachable from the source computes the function $\bigvee_{i \in S(v)} z_{i}$ where $S(v) \subseteq [n]$ is the set of indices of all $z$-variables that do not precede $v$ in the topological order.
    \item[\eqref{item:z-is-1}] An $\oaBP$ computing the function $\gamma'_n|_{\vec{z}=1^n} (\vec{x}, \vec{y}) = \bigvee_{i \in [n]} (x_i \lor y_i)$ is a Hamiltonian path terminating in the $0$-sink with edges of the form $x_i = 0$, $y_i = 0$ going to the next node in the path, and the edges $x_i=1$, $y_i=1$ going to the $1$-sink.
    \item[\eqref{item:z-is-0}] In an $\oaBP$ computing the function $\gamma'_n|_{\vec{z}=0^n} (\vec{x}, \vec{y}) \equiv 0$ the $1$-sink is not reachable from the source.
\end{itemize}
\end{lemma}
\begin{proof}
    The structural property for \eqref{item:z-is-0} is obvious: if there is a path from the source to the $1$-sink, the function is not identically zero.

    The structural property for \eqref{item:z-is-1} is also easy to see. The query complexity of $\gamma'_n|_{\vec{z}=1^n}$ is $2n$, so the $\oaBP$ must have a $2n$-edge path. The edges in these path are the $0$-edges and the path terminates in the $0$-sink. Suppose an edge of the form $y_i=1$ or $z_i=1$ does not go to the $1$-sink, say it goes from a node $u$ to a node $v$, then the path from the source to $u$ according to the Hamiltonian path, then following the edge $uv$ and then continuing to the $0$-sink by the Hamiltonian path contradicts the fact that the $\oaBP$ computes the disjunction of all variables.

    Now let us show the structural property for \eqref{item:x-is-1}. The source computes the function $\bigvee_{i \in [n]} z_i$ by definition. For any node $z_i$ computing $\bigvee_{j \in S} z_j$ for a set $S \subseteq [n]$ we have that the edge $z_i=0$ goes to the node computing $\bigvee_{j \in S \setminus \{i\}} z_j$. For any node $y_i$ that is reachable from the source its edges go to the nodes that compute the same function since otherwise the function computed by the source depends on $y_i$. Then take any node $z_i$. Since the function $\gamma'_n|_{\vec{x}=1^n}$ depends on all bits of $\vec{z}$, this node is reachable from the source. Hence, by the observations above, it computes a function of the form $\bigvee_{j \in S} z_j$ with $i \in S$.

    Observe that every path from the source to $z_i$ must visit all nodes from $\{z_1, \dots, z_n\}$ that precede $z_i$ in the topological order. Indeed, if there is a path from the source to $z_i$ that does not visit $z_j$ that precedes $z_i$ in the topological order, then the $\oaBP$ returns $0$ given the assignment $e_j = 0^{j-1}10^{n-j}$, which is a contradiction. Hence, if the node $z_i$ computes $\bigvee_{j \in S_i} z_j$ and the node $z_k$ computes $\bigvee_{j \in S_k} z_j$, where $z_i$ precedes $z_k$ in the topological order, then $S_i \supseteq S_k$. Moreover $S_i \ni i$ and $S_k \not\ni i$ so $S_i \neq S_k$. Therefore, the sets $S_1, \dots, S_n$ form a chain by inclusion and are all different, which implies the property given that $S_i$ does not contain indices of nodes that precede $z_i$ in the topological order.
\end{proof}

\subsubsection{Proof of \texorpdfstring{\cref{lem:order}}{Lemma order}}
We derive \cref{lem:order} from the following two lemmas.
\begin{lemma}\label{lem:adjacent-mapping} Let $B$ be an $\oaBP$ computing $\gamma'$.
    Then for every topological ordering of the nodes of $B$ there exists a mapping $\pi\colon [n] \to [n]$ such that $x_{\pi(i)}$ neighbors $y_i$ in the topological order. Additionally, the following property holds
    \begin{enumerate}[label=\emph{($\dagger$)}]
\item\itshape \label{it:x-to-sink}
 if $x_{\pi(i)}$ precedes $y_i$ and $z_i$ in the topological order, then the edge $x_{\pi(i)}=1$ does not go to 1-sink.
\end{enumerate}
\end{lemma}
\begin{lemma}
\label{lem:forbidden-subsequence}
Suppose an $\oaBP$ $B$ computes $\gamma'$.
Suppose $\ell_1, \dots, \ell_{3n}$ are the labels of the nodes in some topological ordering of $B$. Then for every $i \in [3n-4]$ and every $a,b,c \in [n]$ we have $\{\ell_i, \ell_{i+1}, \dots, \ell_{i+4}\} \neq \{z_a, y_a, x_b, y_c, z_c\}$.
\end{lemma}

\begin{proof}[Proof of \cref{lem:order}]
    First we show that \cref{lem:forbidden-subsequence} guarantees that $\pi$ given by \cref{lem:adjacent-mapping} is a bijection between $\{y_1, \dots, y_n\}$ and $\{x_1, \dots, x_n\}$.

    Let $\prec$ denote the topological order of the labels of the nodes in $B$. Consider a modified version of $\pi$ that we denote $\pi'$ such that whenever possible $y_i \prec \pi'(y_i) \prec z_i$ or $z_i \prec \pi'(y_i) \prec y_i$. Observe that this modification never introduces additional violations of injectivity since $y_j$ cannot neighbor a node between $y_i$ and $z_i$ in the topological order for $j \neq i$.

    Now suppose $\pi'(y_i) = \pi'(y_j)$ for some $i \neq j$. WLOG assume that $i=1$ and $j=2$. Then, by construction of $\pi'$ and $\pi$ there exists $k$ such that $\ell_k = y_1, \ell_{k+1} = \pi'(y_1), \ell_{k+2} = y_2$. Notice that $\pi'(y_1)$ must be a single element of $\{x_1, \dots, x_n\}$ that neighbors $y_1$ and $y_2$, thus $\ell_{k-1} = z_1$ and $\ell_{k+3} = z_2$. Then we get a contradiction with the statement of \cref{lem:forbidden-subsequence}. Hence $\pi'$ is an injective mapping. As $\pi'$ is an injective mapping that maps each element of a set of the size $n$ into an element of a set of the size $n$, we get that $\pi'$ is a bijection.

    By \cref{lem:or-of-ands-structure} $\ell_1, \dots, \ell_{3n}$ has a subsequence of form $a_1, b_1, \dots, a_n, b_n$ where $\{a_i, b_i\} = \{y_{\tau(i)}, z_{\tau(i)}\}$ for a permutation $\tau\colon [n] \to [n]$.  Observe that $\{\ell_1, \ell_2, \ell_3\}$ must contain $x_i$ for $i \in [n]$, otherwise $\{\ell_1, \ell_2, \ell_3\} = \{a_1, b_1, a_2\}$, contradicting \cref{lem:adjacent-mapping} as $y_{\tau(1)}$ does not neighbor $\pi(y_{\tau(1)})$.

    If $\{\ell_1, \ell_2, \ell_3\}$ contains more than one element of $\{x_1, \dots, x_n\}$ then $\{\ell_1, \ell_2, \ell_3\} \ni y_{\tau(1)}$ by surjectivity of $\pi'$. Moreover, $\ell_3 = x_b$ and either $\ell_1 = x_a$, $\ell_2 = y_{\tau(1)}$ or $\ell_1 = y_{\tau(1)}$, $\ell_2 = x_a$ for some $a,b \in [n], a\neq b$. Then by \cref{lem:or-of-ands-structure} we get $\ell_4 \in \{x_1, \dots, x_n, z_{\tau(1)}\}$. Thus $x_b$ does not belong to the image of $\pi$, which is a contradiction. Therefore, $\{\ell_1, \ell_2, \ell_3\} = \{\pi'(y_{\tau(1)}), y_{\tau(1)}, z_{\tau(1)}\}$. Iteratively applying this procedure to $\{\ell_{3k+1}, \dots, \ell_{3n}\}$ and $\pi'$ restricted to $\{y_{\tau(k+1)}, \dots, y_{\tau(n)}\}$ for $k \in [n]$ concludes the proof of the theorem.
\end{proof}

\subsubsection{Proof of \texorpdfstring{\cref{lem:adjacent-mapping}}{adjacent mapping}}
    Let $B_\alpha$ for $\alpha \in \{0,1\}$ be the diagram obtained from $B$ by contracting all edges of form $x_i = \alpha$ and removing all edges of form $x_i = 1 - \alpha$. Then the diagram $B_0$ computes $\gamma'|_{\vec{x} = 0^n}$. Then from \Cref{lem:or-of-ands-structure} and \eqref{item:x-is-0} we get that $B_0$ consists of a path of length $2n$ that is split into consequent blocks with two nodes in each,  where the labels of the nodes in each block are $\{y_i, z_i\}$ for $i \in [n]$.

    Consider an arbitrary block in $B_0$ and let $u$ and $v$ be the nodes of $B$ corresponding to the nodes in the block, such that $u$ precedes $v$ in the topological order of $B$.

    \paragraph{\boldmath Case 1: $u$ and $v$ are labeled with $z_i$ and $y_i$ respectively.} $u$ must still be present in $B_1$ and be reachable from the source (say, with the assignment $\vec{x} = \vec{y} =1^n;\, \vec{z} = 0^n$). Then consider edges in $B_0$ labeled with $z_i = 1$ and $y_i = 0$: $uv$ and $vw$, where $w$ is the first node in the subsequent block of $B_0$, these are uniquely defined since $B_0$ is an $\oaBP$. We claim that at least one of these edges does not appear in $B$ i.e. it is present in $B_0$ due to a contraction.

    Assume towards a contradiction that both of the edges are present in $B$. Then let us trace a path from the source of $B$ to the $0$-sink according to the assignment $\vec{x} =1^n;\, \vec{y} = 0^n;\, \vec{z} =e_i$. This path must contain $u$ and $w$, as $z_i$-node, corresponding to $u$ should be queried on substitution $x=0^n$, and then, after assigning $z_i=1$ and $y_i=0$ this path leads to the node $w$.
    %Now note, if we'd consider a part of $B_1$ starting at the $w$-node, we would get an $\oaBP$ $B_1'$, which wouldn't depend on the value of $z_i$ anymore, but would depend on values of $z$s which are going after $z_i$ in the topological order.
    By \cref{lem:structures} for \eqref{item:x-is-1} the node $w$ in $B_1$ and our choice of $w$, computes $\bigvee_{j \in S} z_j$. Since $z_i$ precedes $w$ in the topological order, $S \not\ni i$. Hence, we get $B(1^n,0^n,e_i) = 0$ which with value $\gamma'(1^n, 0^n, e_i) = 1$ leads to a contradiction.

    Now, since at least one of the edges $uv$ and $vw$ does not appear in $B$, but appears in $B_0$, the node $v$ in $B$ is adjacent to a node $r$ labeled with an element of $\{x_1, \dots, x_n\}$. Since $uv$ and $vw$ belong to the Hamiltonian path in $B_0$, $r$ must neighbor $v$ in the topological order. Then define $\pi(y_i) = x_j$ where $x_j$ is the label of $r$. Notice that here the node $x_{\pi(i)}$ never precedes both $y_i$ and $z_i$, so the property \emph{\ref{it:x-to-sink}} is satisfied.

    \paragraph{\boldmath Case 2: $u$ and $v$ are labeled with $y_i$ and $z_i$ respectively.}
    First, suppose $y_i$ and $z_i$ are the first two nodes in $B_0$. By \cref{lem:or-of-ands-structure} the edge $y_i=0$ in $B_0$ goes to the first node of the next $yz$-pair, let us denote it by $w$. The edge $y_i=1$ goes to the $z_i$-node in $B_0$. Now, assume, towards a contradiction, that the $y_i=0$-edge in $B$ goes to $w$ as well. Then, a path in $B$ corresponding to an assignment $s=(1^n, 0^n, e_i)$ skips the node $z_i$, as $y_i$ is the source of this $\oaBP$, and $w$ is going after $z_i$ in the topological order. Similarly to the proof of Case 1, we get a contradiction with the fact that the subdiagram of $B_1$ rooted in $w$ computes the function $\bigvee_{j \in S} z_j$ with $S \not\ni i$, which is $0$ on $s$, hence $B(s) = B_1(0^n, e_i) = 0$, whereas $\gamma'(s)=1$.

    Now, assume there is at least one other $yz$-pair queried before $y_i$ and $z_i$.
    Let $y_j$ and $z_j$ be the pair that precedes $y_i$ and $z_i$ in the topological order of $B$.
    Note, that in $B_0$ both edges $y_j=0$ and $z_j=0$ go to $y_i$. And $y_i=0$ edge in $B_0$ goes to a node $w$ (skipping $z_i$), which is the first node in the next $yz$-block.

    \begin{claim}
    \label{claim:non-existence-of-an-edge}
    One of the edges among $z_j=0$ and $y_i=0$ differ in $B_0$ and $B$ (that is, the edge with this label connects different pairs of nodes).
    \end{claim}
    \begin{proof}The proof of this claim repeats a similar argument in the Case 1.
    Assume these two edges in $B$ are the same as in $B_0$.
    Consider a substitution $s=(1^n, 0^n, e_i)$. As $B_1$ computes the disjunction of $z$, a path $p$ corresponding to $s$ in $B$ should query all bits of $\vec{z}$ preceding $z_i$ in the topological order. Hence, $p$ goes through $z_j$, and $z_j=0$ to $y_i$. Now, as $y_i=0$ skips $z_i$  and goes to $w$, which by \cref{lem:structures}  computes $\bigvee_{t \in S} z_t$ in $B_1$. Analogously to the Case 1, $S \not\ni i$, so we get that $B(s)=0$, which contradicts $\gamma'(s)=1$. Thus, one of the edges $z_j=0$ or $y_i=0$ in $B$ leads to a different node compared to $B_0$.\end{proof}

    \paragraph{\boldmath Case 2a: the edge $z_j=0$ goes to $y_i$.} By \cref{claim:non-existence-of-an-edge} $y_i=0$ does not go to $w$. But in this case it goes to a node $r$ labeled by an $x$-variable in $B$. Then as $\oaBP$ $B$ queries all bits in $\vec{z}$ on the assignment $(1^n, 0^n, 0^n)$, we get that $r$ goes before $z_i$ in the topological order. So we get that $y_i$ neighbors an $x$-node in the topological order.

    \paragraph{\boldmath Case 2b: the edge $z_j=0$ does not go to $y_i$.}
     Let $x_\ell$ be the endpoint of $z_j = 0$ in $B$ (notice that $z_j=0$ can not go to $y_j$ since then it would go there in $B_0$ as well). Consider the path from $x_\ell$ to $y_i$ in $B$ according to the assignment $\vec{x} = 0^n$, by the construction of $B_0$ all nodes on this path are labeled with $x$-variables. Let $x_k$ be the last node on this path before $y_i$. We define $\pi(i) = k$. Now it remains to show the following:
     \begin{enumerate}[noitemsep]
         \item $x_k$ is the immediate predecessor of $y_i$ in the topological order. \label{item:predecessor}
         \item The property \emph{\ref{it:x-to-sink}}: $x_k=1$ in $B$ does not go to the $1$-sink. \label{item:not-1-sink}
     \end{enumerate}
      Let us first prove \cref{item:predecessor}. Consider the program $B|_{\vec{z}=1^n}$. By \cref{lem:structures} for \eqref{item:z-is-1}, $B|_{\vec{z}=1^n}$ consists of a single Hamiltonian path of $0$-edges. Then, consider the edge $x_k=0$ in $B$. This edge goes to $y_i$, so $x_k y_i$ is present in $B|_{\vec{z}=1^n}$, so in topological order there are no $x$-node or $y$-node between $x_k$ and $y_i$. Given that $z_j$ precedes $x_k$ in the topological order, $x_k$ is the immediate predecessor of $y_i$.

      Now let us prove \cref{item:not-1-sink}. Consider an assignment $\vec{x}=0^n$, $\vec{y}=1^n$ and $\vec{z}=0^n$. Let us trace the path corresponding to this assignment from the root of $B$ to $z_j$. Notice that $z_j$ is reachable by this path since it is reachable in $B_0$ by the path according to $\vec{y}=1^n$, $\vec{z}=0^n$. After $z_j$ let us follow the edge $z_j=0$ to $x_\ell$ and then to $x_k$ following the $0$-edges. Suppose that $x_k=1$ goes to the $1$-sink. Let $p$ be the path from the root to the $1$-sink that we have traced. Then consider the program $B|_{\vec{z}=0^n}$. Let $p|_{\vec{z}=0^n}$ be the result of contraction of the edges of form $z_t=0$ in $p$. $p|_{\vec{z}=0^n}$ exists in $B_{\vec{z}=0^n}$, hence the $1$-sink is reachable in this program. This contradicts the fact that $\gamma'|_{\vec{z}=0^n}$ is identically zero.

    \begin{figure}
         \centering
         \begin{tabular}{cc}
             \begin{tikzpicture}[node distance=17mm, auto, every node/.append style={rectangle, draw=black},>=stealth, shorten >=1mm, shorten <=1mm]
        %\node(onesink) {$1$-sink};  \node[left of=onesink] (zerosink) {$0$-sink};
        \node (zc) {$z_c$};
        \node[above of=zc] (yc) {$y_{c}$};
        \draw[->] (yc) to[bend right] node[draw=none] {1} (zc);
        \node[above of=yc] (xb) {$x_{b}$};
        \draw[->] (xb) to[bend right] node[draw=none] {0} (yc);
        \node[above of=xb] (ya) {$y_{a}$};
        \draw[->] (ya) to[bend right] node[draw=none] {0} (xb);
        \node[left of=ya] (onesink) {$1$-sink};
        \draw[->] (ya) to node[draw=none] {1} (onesink);
        \node[above of=ya] (za) {$z_{a}$};
        \draw[->] (za) to[bend right] node[draw=none] {1} (ya);
        \draw[dashed, ->] (za.east) to[bend left] node[draw=none] {0} (xb.east);
        \draw[dashed, ->] (za.east) to[bend left] node[draw=none] {0} (yc.east);
        \end{tikzpicture} &
        \begin{tikzpicture}[node distance=17mm, auto, every node/.append style={rectangle, draw=black},>=stealth, shorten >=1mm, shorten <=1mm]
        \node (zc) {$z_c$};
        \node[above of=zc] (yc) {$y_{c}$};
        \draw[->] (yc) to[bend right] node[draw=none] {1} (zc);
        \node[above of=yc] (xb) {$x_{b}$};
        \draw[->] (xb) to[bend right] node[draw=none] {0} (yc);

        \draw[->] (xb) to node[draw=none] {1} (onesink);
        \node[above of=xb] (ya) {$y_{a}$};
        \draw[->] (ya) to[bend right] node[draw=none] {0} (xb);
        \node[left of=ya] (onesink) {$1$-sink};
        \draw[->] (ya) to node[draw=none] {1} (onesink);
        \node[above of=ya] (za) {$z_{a}$};
        \draw[->] (za) to[bend right] node[draw=none] {1} (ya);
        \draw[->] (za.east) to[bend left] node[draw=none] {0} (yc.east);
        \end{tikzpicture}
         \end{tabular}
        \caption{Fixed edges between $\{z_a, y_a, x_b, y_c, z_c\}$. }
        \label{fig:first-round-labels}

    \end{figure}
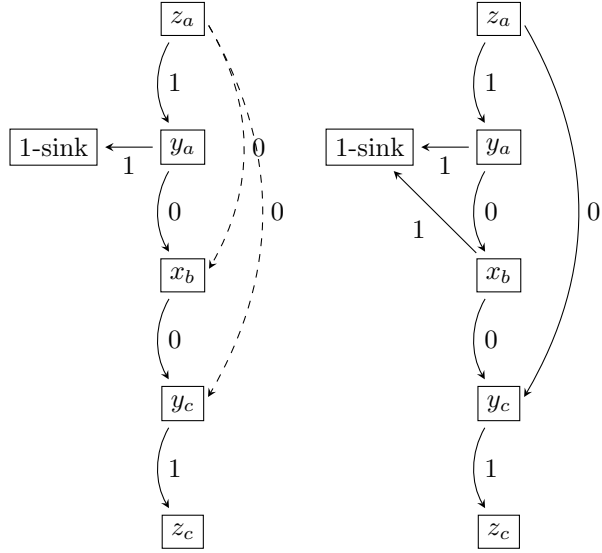
\subsubsection{Proof of \texorpdfstring{\cref{lem:forbidden-subsequence}}{forbidden subsequence}}
    We are going to gradually derive the structure of the edges going out of the nodes $\{z_a, y_a, x_b, y_c, z_c\}$.

    By \eqref{item:z-is-1} we have $\gamma'|_{\vec{z} = 1^n} = \bigvee_{i \in [n]} x_i \lor y_i$, so by \cref{lem:structures} the edges of from $x_i = 0$ and $y_i = 0$ must either go to a variable in $\{z_1, \dots, z_n\}$ or to the next variable in $\{x_1, \dots, x_n, y_1, \dots, y_n\}$ in the topological order. Hence $y_a = 0$ leads to $x_b$, $x_b = 0$ leads to $y_c$.

    By \eqref{item:x-is-0} we have $\gamma'|_{\vec{x}=0^n} = \bigvee_{i \in [n]} z_i \land y_i$, the diagram $B|_{\vec{x}=0^n}$ has structure defined by \cref{lem:or-of-ands-structure}. Then $z_a = 1$ leads to $y_a$ and $y_c = 1$ leads to $z_c$ and $z_a = 0$ leads either to $x_b$ or $y_c$.
    Now observe that the edge $y_a=1$ goes to the $1$-sink by the structure of the diagram $B|_{\vec{z}=1^n}$. If it does go to $x_b$ the assignment $\vec{x} = 0^n$, $\vec{y} = 1^n$ and $\vec{z} = e_a$ is evaluated to $0$ by $B$, which is a contradiction with \eqref{item:z-is-1}. See \cref{fig:first-round-labels} for the edges we identified so far.

    Now consider an edge $x_b = 1$. By \cref{lem:structures} for \eqref{item:z-is-1}  it must go to the $1$-sink in the diagram $B|_{\vec{z}=1^n}$, so in $B$ it goes either to $\{z_1, \dots, z_n\}$ or to the $1$-sink. Recall that by \eqref{item:x-is-1} we have $B|_{\vec{x}=1^n} = \bigvee_{i \in [n]} z_i$. Suppose $x_b = 1$ goes to $z_i$ for $i \in [n]$. Then consider the assignment $\vec{x} = 1^n$, $\vec{z} = e_a$, $\vec{y} = 0^n$. A path $p$, corresponding to it passes through the node $z_a$ by the structure of $B|_{\vec{x}=1^n}$ given by \cref{lem:structures}. Hence, $p$ then leads to $x_b$ and to $z_i$, which computes $\bigvee_{t \in S} z_t$ for some $S \subseteq [n] \setminus \{a\}$. Thus $B(1^n, 0^n, e_a) = 0$, yet $\gamma'(1^n, 0^n, e_a) = 1$. Hence $x_b = 1$ goes to the $1$-sink.

    Now we can pinpoint the endpoint of the edge $z_a = 0$. Suppose it goes to $x_b$. Then we get the contradiction with the assignment $\vec{x} = 1^n$, $\vec{y} = 0^n$ and $\vec{z} = 0^n$. $B(1^n, 0^n, 0^n) = 1$ since we go from $z_a$ to $x_b$ and then to the $1$-sink, yet $\gamma(1^n, 0^n, 0^n) = 0$.

    Then consider the edge $y_c = 0$, we claim that it goes to $z_c$.  Otherwise consider the assignment $\vec{x} = 1^n$, $\vec{y} = 0^n$ and $\vec{z} = e_c$, $B(1^n, 0^n, e_c) = 0$ since the path of this assignment passes through $z_a$, then goes to $y_c$ and then to a node strictly below $z_c$ that by \cref{lem:structures} computes $\bigvee_{i \in S} z_i$ for $S \subseteq [n] \setminus \{c\}$. Thus we get a contradiction with $\gamma'(1^n, 0^n, e_c) = 1$.

    Then we get the final contradiction as follows: $\gamma'|_{\vec{x}=0^n}$ depends on $y_c$, but $B|_{\vec{x}=0^n}$ does not, in other words $B(0^n, e_c, e_c) = B(0^n, 0^n, e_c)$, yet $\gamma'(0^n, e_c, e_c) \neq \gamma'(0^n, 0^n, e_c)$.\qedhere

\subsection{Structure of \texorpdfstring{$\oaBP$}{oaBP} Computing \texorpdfstring{$\gamma_G$}{gamma of G}: proof of \texorpdfstring{\cref{theorem:reduction-correctness-oaBP}}{theorem}}
\label{sec:gamma-G-structure}
\noindent
The theorem follows from the following lemmas, utilizing the properties \eqref{item:y-0-z-1-0}-\eqref{item:edges} of $\gamma_G$.
\begin{lemma}\label{lemma:oaBP_satisfies_gamma_properties_5-6}
For every $\oaBP$ $B$ computing $\gamma_G$, the permutation $\defperm{B}$ defined by $B$ maps elements of $[n/2]$ into $[n/2]$, and elements of $\{n/2+1, \dots, n\}$ into elements of $\{n/2+1, \dots, n\}$.
\end{lemma}

\begin{lemma}\label{lemma:oaBP_satisfies_gamma_property_7}
For every $\oaBP$ $B$ computing $\gamma_G$, the permutation $\defperm{B}$ defined by $B$ corresponds to an independent set in $G$, that is $\{(i, \defperm{B}(i)) \mid i \in [n]\}$ is independent in $G$.
\end{lemma}

We start by deriving the theorem from the lemmas.

\begin{proof}[Proof of \cref{theorem:reduction-correctness-oaBP}]
From \cref{lemma:oaBP_satisfies_gamma_properties_5-6} and \cref{lemma:oaBP_satisfies_gamma_property_7} we get that an $\oaBP$ computing $\gamma_G$ defines a permutation on $[n]$ which maps elements of the first/second half of $[n]$ to the elements of the first/second half of $[n]$ respectively, and for every edge it contains at most one of the vertices of a BPIS instance. Hence, the permutation defined by any $\oaBP$ computing $\gamma_G$ is a correct certificate for a $\BPIS$ instance. Thus an $\oaBP$ computing $\gamma_G$ exists iff there is a correct permutation defining independent set in an input $\BPIS$ instance.
\end{proof}

\subsubsection{Proofs of \texorpdfstring{\cref{lemma:oaBP_satisfies_gamma_properties_5-6}}{lemma} and \texorpdfstring{\cref{lemma:oaBP_satisfies_gamma_property_7}}{lemma}}
Both of the lemmas follow from a helper lemma about the structure of $\oaBP$ for $\gamma'$:
\begin{lemma}\label{lemma:need_two_ones_to_satisfy}
Let $B$ be an $\oaBP$ computing $\gamma'$. Then, if a path $p$ in $B$ goes to the $1$-sink from one of the nodes in a triplet $x_i, z_j, y_j$, then $p$ sets at least two variables among $x_i, z_j, y_j$ to $1$.
\end{lemma}

We defer the proof of \cref{lemma:need_two_ones_to_satisfy} to \cref{sec:proof_two_ones_lemma}. Now we are ready to prove the lemmas.

\begin{proof}[Proof of \cref{lemma:oaBP_satisfies_gamma_properties_5-6}]
 First, as $\gamma_G$ is a restricted version of $\gamma'$, we get that $\defperm{B}$ is well-defined by \cref{lem:order}. Moreover, we know that the topological order of $B$ has the form $a_1,b_1,c_1, a_2, b_2, c_2, \dots, a_{n}, b_{n}, c_{n}$, where for every $h \in [n]$ there exist $k, \ell \in [n]$, such that $\{a_h, b_h, c_h\} = \{x_k, z_\ell, y_\ell\}$, where $\defperm{B}(\ell) = k$.

Let $x_i, z_j, y_j$ form a consequent $xyz$-triplet in the $\oaBP$, with $i = \defperm{B}(j)$. Assume towards a contradiction that $i$ and $j$ are from different halves of the set $[n]$. WLOG let $i \in [n/2]$ and $j \in \{n/2+1,\dots, n\}$.

Consider a substitution $s$, in which $z=1^{n/2}0^{n/2}, y=0^{n}, x_i=e_i$. Consider a path $p$ consistent with $s$. By property \eqref{item:edges} it terminates in the $1$-sink. Now, consider the last $xyz$-triplet, that $p$ goes through before getting to the $1$-sink. This could either be $x_i, z_j, y_j$-triplet, or $x_k, z_\ell, y_\ell$, such that $i \neq k; j \neq \ell$. In both cases, at most one of the variables in such $xyz$-triplets is set to $1$. Hence, we get a contradiction with \cref{lemma:need_two_ones_to_satisfy}. The proof for the $xyz$-triplets in the last $n/2$-elements of $[n]$ goes similarly.
\end{proof}
\begin{proof}[Proof of \cref{lemma:oaBP_satisfies_gamma_property_7}]
Consider an edge connecting $(j,k)$ with $(j', k')$ in $G$. Assume towards a contradiction that $\defperm{B}$ maps an element $j$ to $k$, and $j'$ to $k'$. Then, there are two triplets $x_j, y_{k}, z_{k}$, and $x_{j'}, y_{k'}, z_{k'}$ such that within each $xyz$-triplet the elements are layered consequently in $B$.
Consider a substitution $s$, which sets $x_j = x_{j'}=0$, all other bits of $\vec{x}$ to $1$, $z_k = z_{k'} = 1$, all other bits of $\vec{z}$ to $0$, and all bits of $\vec{y}$ to $0$. $\gamma|_{s}=1$, hence there is a path $p$ in $B$ consistent with $s$ which goes from the source to 1-sink. Consider a triplet $x_a, y_b, z_b$, which is queried the last on the path $p$. Similarly to the previous case we consider two cases, if this triplet is $x_j, y_{k}, z_{k}$ (or $x_{j'}, y_{k'}, z_{k'}$, proof for which is analogous) or not. In the latter case, by our assumptions on the structure of $B$, all $a \neq j$, $a \neq j'$, $b \neq k$, and $b \neq k'$.

If $a=j$ and $b=k$, then, by our choice of $s$ and this triplet, we get that $p$ goes to 1-sink after substituting $x_j=0, z_k=1, y_k=0$. Which is a contradiction with \cref{lemma:need_two_ones_to_satisfy}.

Finally, assume either $a \neq j$ and $b \neq k$. Then there is a path in $B$ from the nodes $x_a, y_b, z_b$ directly to 1-sink which is consistent with a substitution $x_a=1, z_b=0, y_b=0$. Thus, again, leads to a contradiction with \cref{lemma:need_two_ones_to_satisfy}.
\end{proof}

\subsubsection{Proof of \texorpdfstring{\cref{lemma:need_two_ones_to_satisfy}}{the helper lemma}}\label{sec:proof_two_ones_lemma}
Assume the opposite: there exist paths in $B$, which go to the $1$-sink immediately after substituting exactly one $1$-value to variables of an $xyz$-triplet. Consider a path $p$ among these paths, which goes to the $1$-sink from the earliest node according to the topological order of $B$. Let us denote the $xyz$-triplet, from which $p$ goes to the $1$-sink as $x_i, z_j, y_j$, and let $w$ be the last non-sink node in $p$.

\begin{claim}
    $p$ assigns $1$ to $w$.
\end{claim}
\begin{proof}
    Assume the opposite: the last edge in $p$ going to the $1$-sink is $w=0$.

    Consider a case $w=x_i$. We want to show that $x_i=0$-edge cannot lead to the 1-sink. To see this observe that $\oaBP$ $B|_{\vec{z}=1^n}$ is sensitive in all $x$-variables. Hence, in the $\oaBP$ $B|_{\vec{z}=1^n}$ a path consistent with setting $\vec{x}=\vec{y}=0^n$ goes through $x_i$, and goes to the $0$-sink. This contradicts our assumption that the edge $x_i=0$ goes to $1$.

    Now consider a case when $w \in \{y_j, z_j\}$. Consider a path $p'$ in the $\oaBP$ $B|_{\vec{x}=0^n}$ such that for every $k\in [n]$  $p'(z_k) = 1$ and $p'(y_k) = 0$ if $z_k$ precedes $y_k$ in the topological order, and $p'(z_k) = 0$ and $p'(y_k) = 1$ otherwise. By \eqref{item:x-is-0} we have $\gamma'|_{\vec{x}=0^n} = \bigvee_{i \in [n]} (y_i \land z_i)$, so every such $p'$ ends in the $0$-sink. This contradicts the fact that the edge $w=0$ goes to the $1$-sink, since at least one of such paths $p'$ contains this edge. The claim then follows.
\end{proof}

In the remainder of the proof, we show that $w=1$ cannot be the only $1$-value assigned by $p$ to the variables in the last triplet $x_i, z_j, y_j$. Consider the case when $w=y_j$ or $w=z_j$ and the value of $w$ is assigned to $1$ by $p$. Consider an $\oaBP$ $B'=B|_{\vec{x}=0^n}$ corresponding to a partial substitution $\vec{x}=0^n$. We know the structure of $B'$ from \cref{lem:or-of-ands-structure}. For each $zy$-block there is a path to the $1$-sink iff we set both $z_i=y_i=1$. But that contradicts the existence of a path of a form $x_i=0, y_j=0, z_j=1$ or $x_i=0, y_j=1, z_j=0$ from the $x_i, z_j, y_j$ triplet to the $1$-sink. Hence, $w$ could not be $y$ or $z$ variable.

Consider the only remaining case, when $w=x_i$, and $p$ sets $x_i=1$. In the following two claims we show that $p$ should ask either $y_j$ or $z_j$ prior to $x_i$.
\begin{claim}\label{claim:y_or_z_before_x_in_top}
Either $y_j$ or $z_j$ precedes $x_i$ in the topological order.
\end{claim}
\begin{proof}
Assume the opposite, $x_i$ is the first node in the triplet. By \cref{lem:adjacent-mapping} $y_j$ goes either right before or right after $x_i$ in the topological order. Hence, we get that nodes in this triplet appear in the order $x_i, y_j, z_j$ in the topological order. But then, by the property \emph{\ref{it:x-to-sink}} in \cref{lem:adjacent-mapping}, if $x_i$ is the first node of the triplet then the edge $x_i=1$ does not go to the $1$-sink, contradicting the properties of $p$.
\end{proof}
\begin{claim}\label{claim:y_or_z_before_x_on_p}
Path $p$ queries either a value of $y_j$ or $z_j$ before querying $x_i$.
\end{claim}
\begin{proof}
By \cref{claim:y_or_z_before_x_in_top} we know $x_i$ is not the first node of the triplet $x_i, y_j, z_j$ in the topological order. Now assume that $p$ skips querying $yz$-nodes preceding $x_i$ in this triplet and goes directly to $x_i$. Let $w'$ be the last node prior to $x_i$ in $p$. By our assumption $w'$ belongs to an $xyz$-triplet $\{x_b, y_a, z_a\}$ different from $\{x_i, y_j, z_j\}$.

First let us rule out the case $w'=z_a$ or $w'=y_a$. If $p$ assigns $0$ to $w'$, then the edge $w'=0$ in $B|_{\vec{x}=0^n}$ goes to the endpoint of the edge $x_i=0$ in $B$, which contradicts \cref{lem:or-of-ands-structure}, since this edge must go to the first variable of the subsequent $zy$-block, which precedes $x_i$ in the topological order of $B$. If $p$ assigns $1$ to $w'$, then in $B|_{\vec{x}=0^n}$ the edge $w'=1$ goes to the $1$-sink (as it does not go to the node in the same $zy$-block), hence in $B$ the edge $x_i=1$ must go to the $1$-sink, which implies that the function computed by $B$ does not depend on $x_i$, which is a contradiction, since $\gamma'$ does depend on $x_i$.

Hence $w'=x_b$. Let us now show that $p$ assigns $1$ to $x_b$. Assume the opposite: $x_b$ is set to $0$ in $B$. Let us then trace the path corresponding to the assignment $\vec{x}=0^n; \vec{z}=1^n; \vec{y}=0^n$ in $B$ until we reach the node $x_b$. This must happen by \cref{lem:structures} for \eqref{item:z-is-1}. Now by the \cref{lem:or-of-ands-structure} and property \eqref{item:x-is-0} the edge $x_b=0$ must go to the first node of the subsequent $zy$-block, which is a contradiction with the assumption that $x_b=0$ ends in $x_i$. Therefore $p$ assigns $x_b=1$.

Consider an assignment $s$ of $\vec{z}=1^n, \vec{y}=0^n$, $\vec{x}=e_b$. By \eqref{item:z-is-1} we have $\gamma(s)=1$. But by our assumptions an edge $x_b=1$ does not go to the $1$-sink immediately, it goes to $x_i$ first. A node $x_i$ should also be reachable by a substitution $s'$ which assigns $\vec{z}=1^n, \vec{y}=0^n$, $\vec{x}=0^n$, as $\gamma'|_{\vec{z}=1^n}$ is sensitive to the value of $x_i$ by \eqref{item:z-is-1}. But $\gamma(s')=0$. Hence, we get a contradiction with the fact that paths corresponding to assignments $s$ and $s'$ meet in the node $x_i$, and follow the same path in $B$, but one must end in the $0$-sink, and another in the $1$-sink. Therefore, either $z_j$ or $y_j$ node is not queried by $p$.
\end{proof}

 Now, let $r$ be the immediate predecessor of $w$ in $p$. By \cref{claim:y_or_z_before_x_on_p} we know that either $r \in \{y_j, z_j\}$. First, consider the case of $r=z_j$. Consider a substitution $s$ of $\vec{x}=1^n$ and $\vec{z}=0^n$. We know that the $\oaBP$ $B|_{\vec{x}=1^n}$ queries all values of $z$. Hence, $B|_{\vec{x}=1^n}$ contains a path $r=0, w=1$, which goes to the $1$-sink. But, as the assignment $s$ does not satisfy $B|_{\vec{x}=1^n}$, and is consistent with the substitution $r=0, w=1$ we get a contradiction (since the path $s$ passes through $r$).

Finally, to get a contradiction with the case when $r=y_j$ we show the following two claims.
\begin{claim}
If $r=y_j$, then $z_j$ goes before $y_j$ and $x_i$ in the topological order.
\end{claim}
\begin{proof}
Assume the opposite. Let $z_j$ go right after $x_i$ in the topological order. So we get that the order of nodes in this triplet is $y_j, x_i, z_j$. First, note that $y_j=0$ edge goes to $x_i$ as $B|_{\vec{z}=1^n}$ queries all $y$s and $x$s sequentially, and every edge labeled with $0$-edge in $B|_{\vec{z}=1^n}$ goes to either the next node or to the $0$-sink. We also know that $x_i=1$-edge goes to the $1$-sink. Now note that that a path consistent with a substitution $s$ defined as $\vec{x}=e_i, \vec{y}=0^n, \vec{z}=0^n$ should reach node $y_j$. But then this path will follow edge $y_j=0$ to $x_i$ and edge $x_i=1$ to $1$ sink, though $\gamma|_{s}=0$. A contradiction.
\end{proof}
\begin{claim}
If $z_j$ is the first node in a $xyz$-triplet then path $p$ goes through $z_j$.
\end{claim}
\begin{proof}
Consider all possible edges which get to this triplet. Let $x_a, z_b, y_b$ be the triplet right before this one. From the structure of $B|_{\vec{x}=1^n}$ we know that all edges from $z_b$ and $y_b$ either go to a node within this triplet, or the $1$-sink, or to $z_j$.
We also know that $x_a=0$ cannot skip $z_j$, as $x_a$-node is reachable by a path consistent with a substitution $\vec{x}=0^n, \vec{y} = 0^n, \vec{z}=1^n$. Hence if $x_a=0$ skips $z_j$ then $B|_{\vec{x}=0^n}$ would be insensitive to $z_j$.

It now remains to show that $x_a=1$ cannot skip $z_j$ either. In $B|_{\vec{z}=1^n}$ every edge $x_a=1$ should go to $1$-sink. Hence edge $x_a=1$ either goes to $z$-node or to $1$-sink. This finishes the proof.
\end{proof}

From these two claims, we get that $p$ should necessarily go through $z_j$. Hence, by considering $B|_{\vec{x}=1^n}$ we get that a path consistent $z_j=y_j=0$ should either go to a node in the next triplet, or go to the $0$-sink, in case it is the last triplet of $B$. Which contradicts the fact that $p$ after substituting $z_j=y_j=0$ and $x_i=1$ goes to the $1$-sink. This finishes the proof of \cref{lemma:need_two_ones_to_satisfy}.

\section{Corollaries}\label{section:corollaries}
 Multiple corollaries follow from \cref{theorem:reduction-correctness-oaBP}.
 For all classes of branching programs that degenerate to an $\oaBP$ if constrained on having at most $n$ nodes on an $n$-bit input, we get that the corresponding minimization problem is ETH-hard. In particular it holds for general branching programs ($\BP$), ordered binary decision diagrams ($\OBDD$) and read-$k$ branching programs for any $k \ge 1$ ($k$-$\BP$). For definitions of the latter two classes we refer the reader to \cite{Wegener}. Formally we have
\begin{corollary}
\label{cor:hardness-for-class-C}
Let $\mathcal{C}$ be a circuit class such that for every $f\colon \{0,1\}^n \to \{0,1\}$ that depends on all its variables there exists $C \in \mathcal{C}$ with $|C| = n$ computing $f$ if and only if there exists $\oaBP$ $B$ computing $f$.

Then, assuming ETH holds we have that $\mathcal{C}\text{-}{\rm MCSP}|_{s=n}^*$ requires time $2^{\Omega(n \log n)}$.
\end{corollary}
It is easy to see that for $\mathcal{C} \in \{\OBDD, \BP\} \cup \{k\text{-}\BP \mid k \ge 1\}$ this statement holds, so in particular \cref{theorem:main-theorem} is implied.
\begin{proof}[Proof of \cref{cor:hardness-for-class-C}]
    Let $\mathcal{R}$ be the reduction given by \cref{theorem:reduction-correctness-oaBP}. Observe that every partial function in the image of $\mathcal{R}$ can only be extended to a function that depends on all of its variables.

    Hence by the assumption on $\mathcal{C}$, $\mathcal{R}(G) \in \mathcal{C}\text{-}{\rm MCSP}|_{s=n}^*$ iff $\mathcal{R}(G)$ is computable by an $\oaBP$. Therefore, \cref{theorem:reduction-correctness-oaBP} and \cref{theorem:eth-hardness-of-bpis} imply the statement.
\end{proof}

Following on our previous work \cite{GlinskihR22}, in which we show that $\BPIS$ is unconditionally hard for $\oneNBP$s by showing that reduction of Ilango is computable by a $\oneNBP$s, we get that similar $\oneNBP$ lower bound holds for all total minimization problems for branching programs.

\begin{corollary}
For $\mathcal{C}$ as in \cref{cor:hardness-for-class-C}, the size of the smallest $\oneNBP$ that computes $\mathcal{C}\text{-}\MCSP$ is $N^{\Omega(\log\log(N))}$, where $N$ is the input size of $\mathcal{C}\text{-}\MCSP$.
\end{corollary}
\begin{proof}
    By \cite[Theorem~6]{GlinskihR22} we have that the sizes of $\oneNBP$s computing $\mathcal{C}$-$\MCSP^*$ and $\mathcal{C}$-$\MCSP$ are linearly related (the theorem is only stated for vanilla $\MCSP$, but the proof does not use this). By \cite[Threorem~4]{GlinskihR22}, the reduction $\mathcal{R}$ can be implemented in $\oneNBP$, so by \cref{theorem:reduction-correctness-oaBP} $\oneNBP$ size of $\mathcal{C}$-$\MCSP^*$ with input size $N=2^{3n}$ is lower bounded by the $\oneNBP$ size of $\BPIS$ which by \cite[Theorem~3]{GlinskihR22} is $2^{\Omega(n \log n)} = N^{\Omega(\log \log N)}$.
\end{proof}

Another corollary follows from the fact that $\mathcal{R}$ always outputs a partial functions that can be encoded as a $2$-$\BP$.
\begin{corollary}\label{corollary:NP-hardness-2-BP-compression}
It is $\NP$-hard to check, whether any extension of a partial function expressed by a $2$-$\BP$ over the alphabet $\{0,1,*\}$ can be computed with an $\oaBP$.
\end{corollary}
\begin{proof}
    Since $\BPIS$ is an $\NP$-complete problem, it suffices to check that for every $G$ the function $\gamma_G$ can be represented as a $2$-$\BP$ of polynomial size.
    First, notice that a $1$-$\BP$ can distinguish between the cases \eqref{item:x-is-0}-\eqref{item:y-0-z-0-1} and the case that either \eqref{item:edges} is true or the function equals to $*$: we read variables in the order $x_1, y_1, z_1, \dots, x_n, y_n, z_n$ and for each of the patterns among \eqref{item:x-is-0}-\eqref{item:y-0-z-0-1} we maintain if the current prefixes of $\vec{x}$, $\vec{y}$, and $\vec{z}$ coincide with the prefix of the pattern. Thus, at each level of our $1$-$\BP$ we have $2^6$ states. For each of the states at the last level, we choose one of the patterns which is satisfied arbitrary, so we can merge the states into 7: 6 corresponding to the conditions \eqref{item:x-is-0}-\eqref{item:y-0-z-0-1} and one to the case where none of the conditions hold. For each of the cases \eqref{item:x-is-0}-\eqref{item:y-0-z-0-1}, the restricted function can easily be computed by a $1$-$\BP$.

    Now it remains to distinguish the case \eqref{item:edges} from a star with a $1$-$\BP$. In order to do this we split the input into $6$ chunks of length $n/2$ and read then sequentially. For each chunk we maintain the set of possible values among $\{0^{n/2}, 1^{n/2}\} \cup \{e_i \mid i \in [n/2]\} \cup \{\overline{e_i} \mid i \in [n/2]\}$ that the chunk might have given the current prefix. Observe, that the number of potential states is linear. After we read all the chunks we have $O(n^6)$ states at the last level, so we know the only tuple $j,k,j',k'$ that may satisfy the predicate in \eqref{item:edges}. Since the edges are fixed, we can unite each of the states at the last level with the $1$-sink or the $*$-sink.
\end{proof}

Finally, to show the $\coNP$-hardness of compressing $\BP$s we use the following fact.
\begin{proposition}\cite{Tovey84}\label{prop:3-4-SAT_NP_hard}
Let $(3,4)$-$\SAT$ be a language of satisfiable $3$-$\CNF$ formulas in which every variable occurs in at most 4 clauses. $(3,4)$-$\SAT$ is $\NP$-complete.
\end{proposition}

\begin{corollary}
    It is $\coNP$-hard to check whether an input $4$-$\BP$ for a total function can be compressed to a $\BP$ of a constant size.
    \end{corollary}
\begin{proof}
    By \cref{prop:3-4-SAT_NP_hard} we get that $\overline{(3,4)\text{-}\SAT}$ is $\coNP$-hard. Now, we construct a reduction from $\overline{(3,4)\text{-}\SAT}$ to the partial $4$-$\BP$ minimization as follows. For a $4$-$\CNF$ formula $\phi$ in which every variable appears at most $4$ times we  construct a branching program $B$ by consequently adding nodes of each clause. Denote the first clause of $\phi$ as $C_1$, and let $x_i, x_j, x_k$ be the variables in it. We add nodes labeled by $x_i$, $x_j$, and $x_k$ to $B$. We set the node labeled by $x_i$ as a source of $B$. Then, if $x_i$ was negated in $C_1$ we direct a 1-edge from the $x_i$-node to the node labeled by $x_j$, otherwise we direct a 0-edge from $x_i$ to the $x_j$-node. Then we similarly define one of the edges of $x_j$ to $x_k$. Finally, if $x_k$ was negated in $C_1$ we direct its $1$-edge to the 1-sink. Then we add nodes from the next clause of $\phi$, and direct all missing edges for nodes $x_i$, $x_j$, and $x_k$ to the first node corresponding to the next node. We repeat this procedure for all clauses, redirecting all unassigned edges for nodes corresponding to variables in the last clause of $\phi$ to the $0$-sink.

    Now, for any substitution $s$ of values of the variables $B$ outputs $1$ if $\phi$ was not satisfied by $s$, and $0$ otherwise. But note, that if $\phi$ was an unsatisfiable formula than $B$ would had to output $1$ on all inputs. Hence, it could be compressed to a branching program of size 0 (as we count the number of all nodes other than sinks). On the other hand, if $\phi$ is satisfiable and has at least one unsatisfying assignment $B$ has to contain at least one node other than a sink. To account for the fact that the size of $B$ corresponding to $\phi$ that is satisfiable by any assignment is also $0$ (it is just a 0-sink), we change our reduction to first substitute all 0 assignment to $\phi$. If $\phi(\overline{0})=1$ we set $B$ to be a $\BP$ for an identity function on one bit which has the size $1$. Otherwise, construct $B$ as described before.
\end{proof}
 Additionally, as for any 3-$\SAT$ instance in which every variable appears in at most 4 clauses we can easily construct an equivalent $3$-$\BP$ of polynomial size. we get that checking whether this BP can compressed to a $\BP$ of size $1$ allows us to check whether the formula is unsatisfiable.

\section{Future Work}\label{section:future-work}
The main future direction is to strengthen our results to the total version of $\MBPSP$, and to show $\NP$-hardness of $\MBPSP$ (or at least $\MBPSP^*$). We discuss a couple of barriers for improving results in both of these directions as well as additional open questions motivated by our work.

\paragraph{Show polynomial-time oracle reduction from $\MBPSP^*$ to $\MBPSP$.} As noted in \cite{Ilango21} such reductions are very ad hoc for each specific model. Indeed, the reduction that Ilango showed for the Minimum DeMorgan Formula Size Problem in \cite{Ilango21} does not work for branching programs as it relies on the tree-like structure of formulas. The ability of branching programs to reuse the same computation in different branches makes it harder to come up with a suitable reduction. For now we do not even know such a reduction for $\onebp$s but the natural next step would be to reduce $(\onebp)$-$\MCSP^*$ to $(\onebp)$-$\MCSP^*$.

\paragraph{Prove a stronger conditional lower bound.} It would be interesting to find a lower bound stronger than $N^{\Omega(\log\log(N))}$ for the time complexity of $\MBPSP^{*}$ (or $\MFSP^*$). The obstacle here is that we actually show the ETH-hardness of $\MBPSP^*|_{n}$, for which this time complexity is actually tight, since there are $n^{O(n)} = N^{O(\log \log N)}$ potential branching programs that can compute a function. Hence, in order to have stronger lower bounds we need to consider larger size parameters. The current framework requires us to find explicit functions that require $\BP$-size larger than the given size parameter. Hence, the best we can hope for without improving the state-of-the-art $\BP$ lower bounds is to show that $\MBPSP^*|_{n^{2}}$ is ETH-hard. Here the trivial upper bound is $n^{O(n^2)} = N^{O(\log N \log \log N)}$, which while still quasipolynomial, is significantly stronger.

\paragraph{Results for nondeterministic BP minimization.} Currently our hardness result only holds for the problem of minimizing deterministic branching programs. But we believe that similar ideas could be used to show that the problem of minimizing nondeterministic branching programs for partial functions is also ETH-hard. Nondeterministic branching programs ($\NBP$) differ from deterministic ones by an additional type of nodes, called \textit{nondeterministic} nodes, that do not have labels. Such nodes correspond to nondeterministic guessing of a branching program. Note, that if we consider a complexity measure that counts \textit{all nodes in an $\NBP$} other than sinks, then our main result immediately extends to the ETH-hardness of minimizing $\NBP$ and $k$-$\NBP$ for all $k$.

However, the standard complexity measure for $\NBP$s is the minimal number of nodes that query a variable (so all nodes other than sinks and nondeterministic nodes) in an $\NBP$ computing the function. So one of the approaches to showing the ETH-hardness of the Partial Minimum Nondeterministic Branching Program Size Problem could be in showing that existence of the $\oaBP$ with nondeterministic nodes for $\gamma_G$ is equivalent to the existence of the deterministic $\oaBP$ for $\gamma_G$. Namely by adding the nondeterministic node we do not make it easier for an $\oaBP$ to compute $\gamma_G$. Could this result be obtained with an additional careful analysis?

\paragraph{Could we adapt the $\NP$-hardness proof of Hirahara?} Hirahara in~\cite{Hirahara22} showed a randomized reduction from $\NP$-hard problem that is a variation of Minimum Monotone Satisfying Assignment called Collective Minimum Weight Monotone Satisfying Assignment (CMMSA) problem to $\MCSP^{*}$. One of the possible directions of strengthening our result from the ETH-hardness to $\NP$-hardness is by adapting this framework for branching programs. Indeed, the reduction from \cite{Hirahara22} seems to be robust enough to correctly work on all \textit{no}-instances of CMMSA even for branching programs.

But the same reduction does not immediately go through for the \textit{yes}-instances. The main obstacle that we see is in showing an analogue of the Uhlig's lemma \cite{Uhlig74, Uhlig92} for branching programs. This lemma states that in order to compute the direct product of the same function (of arbitrary circuit complexity) on subexponentially many inputs, it is sufficient to use one circuit of size $O(2^{n}/{n})$. Such efficient compression is an important step for constructing a small circuit corresponding to the \textit{yes}-instances of CMMSA.

It is not known whether an analogue of Uhlig's lemma exists for branching programs.
Rao and Sinha \cite{RaoS16} showed a variation of the direct-sum theorem for read-once BPs, hence disproving the analogue of Uhlig's lemma in $\onebp$ setting.
We believe that showing an analogue of Uhlig's lemma for general branching programs in the specific setting required by the Hirahara's proof is a promising direction for showing $\NP$-hardness of $\MBPSP^{*}$.

\section*{Acknowledgments}
The authors are grateful to Rahul Santhanam for suggesting working on connections between meta-complexity and branching programs, and to Mark Bun and anonymous reviewers for very helpful feedback on the preliminary versions of this paper.

Artur Riazanov was supported by the Swiss State Secretariat for Education, Research and Innovation (SERI) under contract number MB22.00026.

Ludmila Glinskih was supported by a Sloan Research Fellowship to Mark Bun. This work was done in part while Ludmila was visiting the Simons Institute for the Theory of Computing.
\bibliographystyle{alpha}
\bibliography{bib}
\end{document}